\documentclass[preprint,12pt]{elsarticle}

\usepackage{amssymb}
\usepackage{amsmath}

\usepackage{graphicx}%
\usepackage{multirow}%
\usepackage{amsfonts}%
\usepackage{amsmath}
\usepackage{amssymb}
\usepackage{amsthm}
\usepackage{mathrsfs}%
\usepackage[utf8]{inputenc}
\usepackage{booktabs}
\usepackage{array}
\usepackage{geometry}
\usepackage{color}
\usepackage{xcolor}
\usepackage{float}
\usepackage{algorithm}
\usepackage{algpseudocode}
\usepackage{mathtools, nccmath}
\usepackage{hyperref}

\newtheorem{theorem}{Theorem}
\newtheorem{lemma}{Lemma}

\DeclareMathOperator*{\argmin}{argmin}
\definecolor{customorange}{RGB}{255,127,0}
\definecolor{ForestGreen}{RGB}{60, 170, 60}

\newcommand{\x}{\textbf{x}}

\journal{Journal of Computational Physics}

\begin{document}

\begin{frontmatter}



\title{Implicit Binarization via Complex Phase Dynamics in Combinatorial Optimization}


\author[inst1]{Khen Cohen\fnref{equal}}
\ead{khencohen@mail.tau.ac.il}

\author[inst2]{Mark Glass\fnref{equal}}
\ead{markglass@mail.tau.ac.il}


\author[inst2]{Meir Feder}

\author[inst1]{Yaron Oz}

\fntext[equal]{These authors contributed equally to this work.}

\affiliation[inst1]{organization={School of Physics and Astronomy, Tel Aviv University},
            city={Tel Aviv},
            country={Israel}}

\affiliation[inst2]{organization={School of Electrical and Computer Engineering, Tel Aviv University},
            city={Tel Aviv},
            country={Israel}}

\begin{abstract}
We introduce a physics-inspired continuous relaxation framework that yields substantially improved solutions for NP-hard combinatorial optimization problems, including Quadratic Unconstrained Binary Optimization (QUBO), binary sparse coding, and planted-solution Ising models. By parameterizing discrete binary variables as continuous wave-like states on the complex unit circle, we inherently smooth highly non-convex energy landscapes. We show that representing binary variables as complex phases reveals an implicit regularization mechanism that promotes convergence toward discrete states. Extracting this mechanism yields significant improvements even within standard real-valued optimization frameworks, using this regularizer explicitly. Empirically, this regularization yields vastly higher ground-state convergence rates than standard real-valued alternatives. Our models achieved zero error in large-scale $160 \times 160$ QUBO tasks under severe noise ($\sigma=0.25$), and outperformed traditional algorithms (OMP and LASSO) in underdefined sparse coding with perfect recovery at $\sigma=0.15$. The solver's robustness was further validated by recovering exact ground-state configurations in $8$ out of $11$ rigorously engineered planted-solution benchmarks.
\end{abstract}


\end{frontmatter}

\section{Introduction}

Combinatorial optimization is a foundational challenge in fields 
ranging from operations research \cite{korte2012combinatorial} to cryptography and statistical physics \cite{mezard2009information}. These problems involve finding an optimal configuration from a finite, discrete set of objects - a task that often proves computationally intractable due to the exponential explosion of the search space. Many of these challenges are classified as NP-hard \cite{karp1972reducibility}, necessitating the development of sophisticated heuristic, continuous, or physics-inspired strategies to find high-quality solutions within reasonable computational limits. The difficulty and fundamental importance of these challenges can be illustrated through three typical problem classes:

\textbf{Quadratic Unconstrained Binary Optimization (QUBO)} has emerged as a unifying mathematical framework for formulating a vast array of combinatorial challenges, including the traveling salesperson problem, Max-Cut, and graph coloring \cite{kochenberger2014unconstrained}. In QUBO, the objective is to minimize a quadratic polynomial over binary variables. While the formulation is mathematically simple, the energy landscape is highly non-convex. Recently, quantum and quantum-inspired approaches have shown promise for addressing QUBO formulations - including adaptations for binary sparse coding \cite{romano2024quantum} - demonstrating superior performance compared to alternative classical methods, particularly in small sample size systems.

\textbf{Sparse coding} stands as one of the most fundamental problems in signal processing, with a wide range of applications, including post-quantum encryption schemes \cite{regev2005lattices}, imaging systems \cite{yoon2023uncoverImaging, bratke2019acceleratedMRI}, communication channels \cite{yu2021sparseComm}, and radar systems \cite{herman2008compressedRadar}. In sparse coding, the goal is to recover an unknown sparse vector from a noisy linear system. This computational challenge becomes particularly acute for binary sparse coding, where the underlying signal takes only discrete values \cite{natarajan1995sparse,davis1997adaptive}. This significantly constrains the feasible solution space and renders traditional convex relaxation approaches inadequate.

\textbf{Planted Solutions} provide a rigorous paradigm for evaluating optimization algorithms by utilizing engineered Ising problems with predefined, verifiable ground states \cite{Hen2015}. Unlike purely random instances where the true global minimum is computationally intractable to verify, planted-solution Hamiltonians deliberately encode a specific target configuration, such as a set of logical constraints or factorization rules. This exact \textit{a priori} knowledge yields an unambiguous metric for success, allowing researchers to definitively ascertain whether an algorithm, such as ours, reached the absolute optimum rather than stalling in a local minimum. Because they generate structured energy landscapes that are significantly more informative than fully random spin glasses, planted constructions - ranging from frustration-based instances \cite{Hen2015} to equation planting \cite{Hen2019} and XORSAT benchmarks \cite{Kowalsky2022} - have become central to systematically evaluating and comparing the performance of both classical and quantum optimization devices.

A widely adopted approach to tackle these combinatorial problems is to apply continuous convex relaxation techniques. For instance, in sparse optimization, the mathematically intractable $\ell_0$-norm is frequently substituted with the $\ell_1$-norm, leading to the Least Absolute Shrinkage and Selection Operator (LASSO) \cite{tibshirani1996regression}. While LASSO provides computational tractability through convex programming, empirical and theoretical evidence demonstrates that it suffers from fundamental limitations in strictly combinatorial scenarios, often exhibiting suboptimal recovery performance when dealing with heavily quantized or discrete signals \cite{candes2008enhancing}. Continuous relaxations often fail to tightly bound the discrete feasible region, resulting in fractional solutions that require heuristic rounding steps, which can severely degrade the quality of the final discrete assignment.

To overcome the limitations of real-valued relaxations, researchers are increasingly mapping discrete variables into higher-dimensional, physically-inspired spaces. Complex numbers, for instance, naturally arise in physical contexts like electromagnetic waves, radar, and quantum mechanics, where phase and magnitude encapsulate coupled dynamics \cite{born_wolf_1999, nielsen_chuang_2010}. Extending this logic, quaternion algebra provides an even higher-dimensional framework \cite{sangwine2000colour, ell2007quaternion, mandic2011quaternion}. By mapping discrete assignments to the phases of complex numbers, one can exploit continuous symmetries and wave-like properties to navigate the optimization landscape more effectively. The complex phase representation constrains the optimization trajectory to a curved manifold, introducing an effective penalty that suppresses deviations from binary states. However, we emphasize that the complex representation itself is not the primary source of improved performance. Rather, we show that representing binary variables as complex phases reveals an implicit regularization mechanism that promotes convergence toward discrete states, and that can be used explicitly.

In this paper, we propose this physics-inspired approach to solve combinatorial optimization problems. We tackle the three fundamental problems mentioned above and demonstrate that by converting the problem into an Ising model formulation and adopting a physical, complex-wave interpretation for each discrete bit, we expose this underlying implicit regularization mechanism. Extracting this mechanism yields improved performance even within standard real-valued optimization frameworks, revealing a substantial performance boost over contemporary alternative methods. Furthermore, we extend this formulation to higher-dimensional algebraic spaces. We mathematically and empirically demonstrate that the 2D complex manifold is both necessary and sufficient to induce this topological regularization.

The remainder of this paper is structured as follows. In Section~\ref{sec:problemsformulation} we define the three optimization problems of interest that we deal with in this work: QUBO, Sparse Coding, and Planted Solutions. In Section~\ref{sec:formulation}, we introduce the theoretical framework, the physics inspired approach and the algorithm (Alg.\ref{alg:training_unified}).
Section~\ref{sec:TestRes} presents the simulation results of all the cases. We then discuss our findings in Section~\ref{sec:discussion}, where we also highlight several open questions and directions for future research. 
This paper is supported by four Appendices with proofs and additional numerical examples.

\section{The Optimization Problems}
\label{sec:problemsformulation}

In this section, we formulate the three underlying optimization problems, explain the relationship between them, and review their computational complexity. The first problem is the Quadratic Unconstrained Binary Optimization (QUBO) characterized by a square measurement matrix. The second problem, sparse coding, is presented as an underdefined variation with a known cardinality of the first problem. Finally, the third problem involves a set of eleven planted-solution Ising models of various sizes with known ground states, which serve as rigorous benchmarks.

\subsection{Quadratic Unconstrained Binary Optimization}
\label{sub:formQUBO}
In our QUBO formulation, the goal is to recover an unknown binary vector $\mathbf{x} \in \{0,1\}^N$ from noisy linear measurements $\mathbf{b}$ governed by the following system:
\begin{equation} \label{eq:qubo_system}
    \mathbf{b} = \mathbf{A}\mathbf{x} + \mathbf{e} \ ,
\end{equation}
where the square measurement matrix $\mathbf{A} \in \mathbb{R}^{N \times N}$ satisfying $\det(\mathbf{A}) = 1$ is known, and $\mathbf{e} \sim \mathcal{N}(0, \sigma^2 \mathbf{I})$ is an unknown random Gaussian noise vector with standard deviation $\sigma$ (e.g., $\sigma = 0.1$). 

For a noise standard deviation equal (or close enough) to zero, the unknown vector $\mathbf{x}$ can be easily recovered by multiplying the measurement vector $\mathbf{b}$ by the inverse of the matrix $\mathbf{A}$ and rounding to the nearest discrete values:
\begin{equation} \label{eq:ml}
    \hat{\mathbf{x}} = \lfloor \mathbf{A}^{-1}\mathbf{b} \rceil .\nonumber
\end{equation}
However, starting from a certain noise threshold, this naive approach no longer yields the correct answer. Finding the optimal solution then requires an exhaustive search to minimize the squared Euclidean distance, which entails exponential computational complexity.
\begin{equation} \label{eq:ml}
    \hat{\mathbf{x}} = \argmin_{\mathbf{x} \in \{0,1\}^N} \|\mathbf{A}\mathbf{x} - \mathbf{b}\|_2^2 \ .
\end{equation}
For instance, for the problem size of $N=160$ evaluated later in this work, the optimal solution would require testing $2^{160}$ different combinations.
Expanding this objective cost function yields $\mathbf{x}^T (\mathbf{A}^T \mathbf{A}) \mathbf{x} - 2\mathbf{b}^T \mathbf{A} \mathbf{x}$, which is mathematically equivalent to the standard QUBO formulation.

\subsection{Sparse Coding Problem}
\label{sub:formSC}
The sparse coding problem can be regarded as a variation of the previously defined QUBO model. It relies on the identical linear measurement system established in Eq. (\ref{eq:qubo_system}), but differs fundamentally in the shape of the matrix and the prior structural knowledge regarding the signal. Here, the problem is underdefined: the measurement matrix $\mathbf{A} \in \mathbb{R}^{M \times N}$ contains significantly fewer rows than columns ($M \ll N$) and satisfies $\det(\mathbf{A}\mathbf{A}^T) = 1$. 

To compensate for the underdefined nature of the system, we are provided with additional information about the target vector. The \textit{cardinality}, denoted by $C$, is defined as the exact number of non-zero elements in the vector, mathematically represented by the $\ell_0$ pseudo-norm ($\|\mathbf{x}\|_0 = C$).

The traditional approach to estimating the correct solution modifies the objective from Eq. (\ref{eq:ml}) by optimizing over two loss components - the standard data fidelity and a sparsity regularizer:
\begin{equation} \label{eq:sparseLoss}
    \hat{\mathbf{x}} = \argmin_{\mathbf{x} \in \{0,1\}^N} \  \|\mathbf{A}\mathbf{x} - \mathbf{b}\|_2^2 + \lambda \|\mathbf{x}\|_0 \ ,
\end{equation}
where the first term is the \textit{Data} term, and the second is the \textit{Sparsity} term. $\lambda$ is a hyperparameter balancing the two objectives. In our specific case, where the target cardinality $C$ is strictly known, $\lambda$ is dynamically tuned to explicitly enforce $\|\hat{\mathbf{x}}\|_0 = C$. 

\subsection{Planted Solution Ising Problem}
\label{sub:formPS}
To benchmark our algorithm, we employ a complementary planted construction derived from binary integer multiplication. We select two prime numbers, $p$ and $q$, whose binary digits are denoted by $p_i, q_j \in \{0,1\}$. We represent these primes and compute their semiprime product $N = pq$ according to the following polynomial expansions:
\begin{gather}
    p = \sum_{i=0}^{n_p-1} p_i \, 2^i, \qquad q = \sum_{j=0}^{n_q-1} q_j \, 2^j, \\
    N = pq = \sum_{k=0}^{n_p+n_q-2} N_k \, 2^k.
\end{gather}

The objective is to construct an Ising Hamiltonian whose ground state exactly encodes the bit strings $(p,q)$ satisfying these multiplicative constraints.
A detailed description of the pipeline that constructs these problems is provided in \cite{hen2026plantedsolutionsatisingbenchmarks}.



\section{Solver Based on Physics Informed Relaxation}
\label{sec:formulation}

In this section, we establish a physics-informed framework for solving the optimization problems presented previously. We do so by formulating the problem as an Ising model and associating it with four distinct algebraic structures for parameterizing the continuous state space: real, complex, spherical, and quaternion. Each model relaxes the discrete binary variables into a higher-dimensional continuous manifold, as illustrated in Fig. \ref{fig:MainDiagram}.

\begin{figure}[H]
    \centering
    \includegraphics[width=1.0\linewidth]{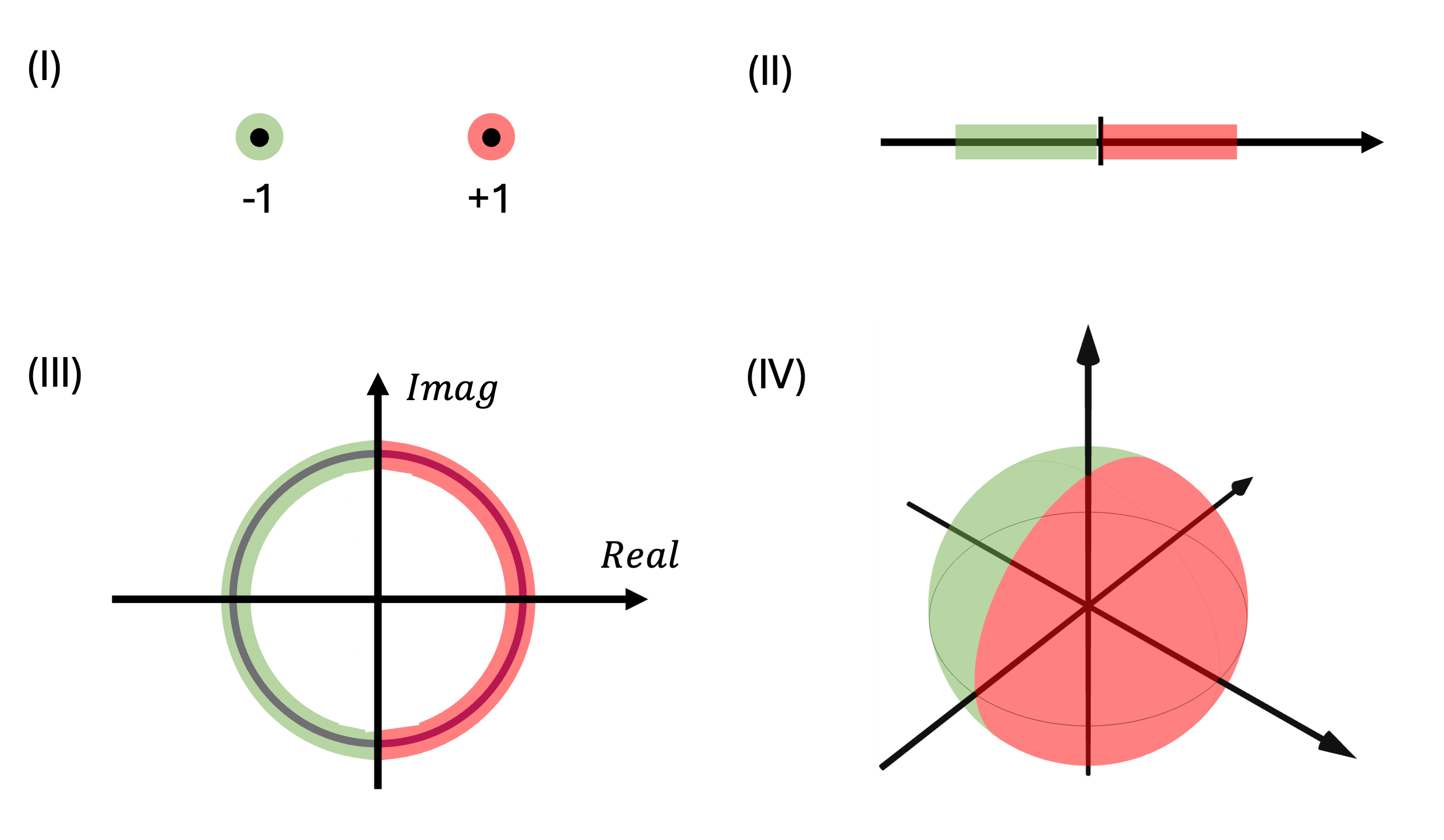}
    \caption{Visualization of continuous relaxation. (I) A binary assignment is used, taking discrete values in $\{\pm1\}$. (II) Optimization is carried out over the real line $\mathbb{R}$; at the end of the procedure, the final discrete value is selected based on its color region. (III) Optimization is performed on the complex manifold $C^1$, followed by the same rounding step based on the color region. (IV) Optimization over a 3D sphere.}
    \label{fig:MainDiagram}
\end{figure}

\subsection{Complex Field Relaxation}
\label{sub:complexRelaxation}
Here, we focus on the QUBO formulation as a case study to analyze the specific effects of the complex representation - we subsequently apply this framework to the remaining problem classes to demonstrate generalized performance improvements. We start by establishing a baseline using the standard real-valued relaxation, where the binary variable is parameterized continuously as $\mathbf{x} = \cos(\boldsymbol{\theta})$ (note it is written in a vectorial notation, i.e. $x_i = \cos(\theta_i)$). This leads to the unconstrained minimization of the data fidelity:
\begin{equation}
\label{eq:optimization}
    \min_{\boldsymbol{\theta} \in \mathbb{R}^N} \|\mathbf{A}\cos(\boldsymbol{\theta}) - \mathbf{b} \|_2^2 \ .
\end{equation}

We extend this formulation by projecting the problem into the complex plane, introducing a unit-modulus complex variable $\mathbf{z} = e^{i\boldsymbol{\theta}} \in \mathbb{C}^N$. The unregularized optimization problem in the complex field then becomes:
\begin{align}
\label{eq:quboLoss}
    \min_{\boldsymbol{\theta} \in \mathbb{R}^N} \|\mathbf{A} e^{i\boldsymbol{\theta}} - \mathbf{b} \|_2^2 
    &= \min_{\boldsymbol{\theta} \in \mathbb{R}^N} \left\| (\mathbf{A}\cos(\boldsymbol{\theta}) - \mathbf{b}) + i (\mathbf{A}\sin(\boldsymbol{\theta})) \right\|_2^2 \nonumber \\
    &= \min_{\boldsymbol{\theta} \in \mathbb{R}^N} \underbrace{\|\mathbf{A}\cos(\boldsymbol{\theta}) - \mathbf{b} \|_2^2}_{\mathcal{L}_{\text{data}}} + \underbrace{\|\mathbf{A}\sin(\boldsymbol{\theta})\|_2^2}_{\mathcal{R}} \ .
\end{align}

The resulting objective function naturally decomposes into the standard real-valued data fidelity term, $\mathcal{L}_{\text{data}}$, and a secondary residual term, $\mathcal{R}$. The key observation is that the complex formulation introduces this additional residual term, which acts as an implicit regularizer that biases the solution toward discrete binary states. Importantly, this regularization can be isolated and applied independently of the complex representation. As shown analytically in \ref{app:complexReg}, provided the matrix $\mathbf{A}$ satisfies standard sub-Gaussian ensemble properties, minimizing $\mathcal{R}$ intrinsically drives this discrete convergence. We note, however, that the effectiveness of this binarization mechanism is tied to the spectral structure of $\mathbf{A}$, and its strength may diminish in the presence of strong column correlations or highly anisotropic singular values. Interestingly, we further demonstrate that projecting the problem into higher-dimensional spaces beyond the complex field yields no additional regularization terms. We further derive the corresponding loss terms in the presence of a LASSO regularizer in \ref{app:LassoReg}.

\subsection{Theoretical Characterization of Implicit Binarization}
\label{sec:binarization}

In this subsection we formalize the binarization mechanism induced by the
residual term in Eq.~(\ref{eq:quboLoss}). We work with the objective
\begin{equation}
\label{eq:total-loss}
L(\boldsymbol{\theta}) \;=\;
\underbrace{\|\mathbf{A}\cos\boldsymbol{\theta} - \mathbf{b}\|_2^2}_{L_{\mathrm{data}}(\boldsymbol{\theta})}
\;+\;
\underbrace{\|\mathbf{A}\sin\boldsymbol{\theta}\|_2^2}_{R(\boldsymbol{\theta})},
\qquad \boldsymbol{\theta}\in\mathbb{R}^N.
\end{equation}
For each $\boldsymbol{\theta}$ define the nearest binary projection in phase
space,
\begin{equation}
\label{eq:projection}
\tilde{\theta}_i \;=\; \pi\cdot\mathrm{round}\!\left(\theta_i/\pi\right),
\qquad
\delta_i \;=\; \theta_i - \tilde{\theta}_i \in\bigl[-\tfrac{\pi}{2},\,\tfrac{\pi}{2}\bigr].
\end{equation}
Since $\cos\tilde{\theta}_i = \pm 1$, the configuration $\tilde{\boldsymbol{\theta}}$
is binary in the original variables. Our goal is to compare
$L(\boldsymbol{\theta})$ to $L(\tilde{\boldsymbol{\theta}})$ and identify the
regime in which $L(\tilde{\boldsymbol{\theta}}) \le L(\boldsymbol{\theta})$,
so that no continuous configuration can improve on its binary projection.
We denote by $\sigma_{\min}$ and $\sigma_{\max}$ the smallest and largest
singular values of $\mathbf{A}$.

\subsubsection{Sharp control of the data-term increment}

The key technical observation is that $\cos$ is \emph{flat} at multiples of
$\pi$, so the data-term increment $|\cos\theta_i - \cos\tilde{\theta}_i|$
is quadratic, not linear, in $\delta_i$. This is precisely the regime where
the binarization mechanism operates, and it makes the regularizer dominate
the data-term sensitivity.

\begin{lemma}[Quadratic flatness of cosine at $\pi\mathbb{Z}$]
\label{lem:quadratic-flatness}
For every $\boldsymbol{\theta} \in \mathbb{R}^N$ and each coordinate $i$,
\begin{equation}
\bigl|\cos\theta_i - \cos\tilde\theta_i\bigr| \;\le\; \sin^2\theta_i.
\label{eq:cos-bound}
\end{equation}
\end{lemma}

\begin{proof}
Since $\tilde\theta_i = k_i\pi$ for some integer $k_i$, we have
$\cos\tilde\theta_i = (-1)^{k_i}$ and $\cos\theta_i = (-1)^{k_i}\cos\delta_i$,
where $\delta_i = \theta_i - \tilde\theta_i \in [-\pi/2,\, \pi/2]$. Hence
\[
\bigl|\cos\theta_i - \cos\tilde\theta_i\bigr| \;=\; 1 - \cos\delta_i.
\]
Applying the double-angle identity in two equivalent forms,
\[
1 - \cos\delta_i \;=\; 2\sin^2(\delta_i/2) \;=\; \frac{\sin^2\delta_i}{2\cos^2(\delta_i/2)}.
\]
On the interval $\delta_i \in [-\pi/2,\, \pi/2]$ we have $|\delta_i/2| \le \pi/4$, so
$\cos^2(\delta_i/2) \ge \cos^2(\pi/4) = 1/2$. Therefore
\[
1 - \cos\delta_i \;\le\; \sin^2\delta_i \;=\; \sin^2\theta_i,
\]
where the final equality uses $|\sin\theta_i| = |\sin\delta_i|$, which follows from
$\theta_i = \tilde\theta_i + \delta_i$ with $\tilde\theta_i \in \pi\mathbb{Z}$.
\end{proof}

\begin{lemma}[Data-term increment]
\label{lem:data-increment}
For every $\boldsymbol{\theta} \in \mathbb{R}^N$,
\begin{equation}
\|\mathbf{A}(\cos\boldsymbol{\tilde\theta} - \cos\boldsymbol{\theta})\|_2^2 \;\le\; \sigma_{\max}^2\,\|\sin\boldsymbol{\theta}\|_\infty^2\,\|\sin\boldsymbol{\theta}\|_2^2.
\label{eq:data-increment}
\end{equation}
Moreover,
\begin{equation}
\bigl|L_{\mathrm{data}}(\boldsymbol{\tilde\theta}) - L_{\mathrm{data}}(\boldsymbol{\theta})\bigr|
\;\le\; \|\mathbf{A}(\cos\boldsymbol{\tilde\theta} - \cos\boldsymbol{\theta})\|_2 \Bigl(2\,\|\mathbf{A}\cos\boldsymbol{\theta} - \mathbf{b}\|_2 + \|\mathbf{A}(\cos\boldsymbol{\tilde\theta} - \cos\boldsymbol{\theta})\|_2\Bigr).
\label{eq:data-difference}
\end{equation}
\end{lemma}
\begin{proof}
By Lemma~\ref{lem:quadratic-flatness}, $(\cos\tilde\theta_i - \cos\theta_i)^2 \le \sin^4\theta_i$
coordinatewise, so
\[
\|\cos\boldsymbol{\tilde\theta} - \cos\boldsymbol{\theta}\|_2^2
\;\le\; \sum_{i=1}^N \sin^4\theta_i
\;=\; \|\sin\boldsymbol{\theta}\|_4^4
\;\le\; \|\sin\boldsymbol{\theta}\|_\infty^2\,\|\sin\boldsymbol{\theta}\|_2^2.
\]
Applying the operator-norm bound $\|\mathbf{A}\mathbf{v}\|_2 \le \sigma_{\max}\|\mathbf{v}\|_2$ yields
\eqref{eq:data-increment}.
For \eqref{eq:data-difference}, set $\mathbf{u} = \mathbf{A}\cos\boldsymbol{\tilde\theta} - \mathbf{b}$ and
$\mathbf{v} = \mathbf{A}\cos\boldsymbol{\theta} - \mathbf{b}$, so that $\mathbf{u} - \mathbf{v} = \mathbf{A}(\cos\boldsymbol{\tilde\theta} - \cos\boldsymbol{\theta})$.
The algebraic identity $\|\mathbf{u}\|_2^2 - \|\mathbf{v}\|_2^2 = (\mathbf{u}-\mathbf{v})\cdot(\mathbf{u}+\mathbf{v})$ together with
Cauchy--Schwarz gives
\[
\bigl|\,\|\mathbf{u}\|_2^2 - \|\mathbf{v}\|_2^2\,\bigr|
\;\le\; \|\mathbf{u}-\mathbf{v}\|_2\,\|\mathbf{u}+\mathbf{v}\|_2
\;\le\; \|\mathbf{u}-\mathbf{v}\|_2\,\bigl(2\|\mathbf{v}\|_2 + \|\mathbf{u}-\mathbf{v}\|_2\bigr),
\]
which is \eqref{eq:data-difference}.
\end{proof}

\subsubsection{Local binarization}

We can now state the main result of this subsection. It says that any
configuration sufficiently close to the binary set is improved by
projecting onto it, with the improvement quantified by the residual norm.

\begin{theorem}[Local binarization]
\label{thm:local-binarization}
Assume $\sigma_{\min} > 0$ (so $\mathbf{A}$ has trivial null space; this covers the square nonsingular and overdetermined regimes). For every $\boldsymbol{\theta} \in \mathbb{R}^N$ with $\sin\boldsymbol{\theta} \ne 0$,
define the threshold
\begin{equation}
\eta(\boldsymbol{\theta}) \;=\; \frac{\sigma_{\min}^2}{2\,\sigma_{\max}\!\left(\sigma_{\max} + 2\,\|\mathbf{A}\cos\boldsymbol{\theta} - \mathbf{b}\|_2 \big/ \|\sin\boldsymbol{\theta}\|_2\right)}.
\label{eq:threshold}
\end{equation}

\begin{enumerate}
\item[(i)] If $\sin\boldsymbol{\theta} \ne 0$ and $\|\sin\boldsymbol{\theta}\|_\infty \le \eta(\boldsymbol{\theta})$, then
\begin{equation}
L(\tilde{\boldsymbol{\theta}}) \;\le\; L(\boldsymbol{\theta}) - \tfrac{1}{2}\,\sigma_{\min}^2\,\|\sin\boldsymbol{\theta}\|_2^2.
\label{eq:descent}
\end{equation}

\item[(ii)] Define
\[
B \;=\; \bigl\{\boldsymbol{\theta} \in \mathbb{R}^N : \sin\boldsymbol{\theta} = 0\bigr\}
\;\cup\; \bigl\{\boldsymbol{\theta} \in \mathbb{R}^N : \sin\boldsymbol{\theta} \ne 0 \text{ and } \|\sin\boldsymbol{\theta}\|_\infty \le \eta(\boldsymbol{\theta})\bigr\}.
\]
Any global minimizer of $L$ on $\mathbb{R}^N$ that lies in $B$ is binary (i.e., satisfies $\sin\boldsymbol{\theta} = 0$).
\end{enumerate}
\end{theorem}

\begin{proof}
\textit{(i)} Since $\sin\tilde{\boldsymbol{\theta}} = 0$, we have $R(\tilde{\boldsymbol{\theta}}) = 0$, so
\[
L(\tilde{\boldsymbol{\theta}}) - L(\boldsymbol{\theta}) \;=\; \bigl(L_{\mathrm{data}}(\tilde{\boldsymbol{\theta}}) - L_{\mathrm{data}}(\boldsymbol{\theta})\bigr) - R(\boldsymbol{\theta}).
\]
Combining \eqref{eq:data-difference} with \eqref{eq:data-increment} gives
\begin{equation}
L_{\mathrm{data}}(\tilde{\boldsymbol{\theta}}) - L_{\mathrm{data}}(\boldsymbol{\theta})
\;\le\; \sigma_{\max}\,\|\sin\boldsymbol{\theta}\|_\infty\,\|\sin\boldsymbol{\theta}\|_2 \Bigl(2\|\mathbf{A}\cos\boldsymbol{\theta} - \mathbf{b}\|_2 + \sigma_{\max}\,\|\sin\boldsymbol{\theta}\|_\infty\,\|\sin\boldsymbol{\theta}\|_2\Bigr).
\label{eq:data-bound}
\end{equation}
Together with the regularizer lower bound $R(\boldsymbol{\theta}) \ge \sigma_{\min}^2 \|\sin\boldsymbol{\theta}\|_2^2$,
the inequality $L(\tilde{\boldsymbol{\theta}}) - L(\boldsymbol{\theta}) \le -\tfrac12 \sigma_{\min}^2 \|\sin\boldsymbol{\theta}\|_2^2$
is implied, after dividing through by $\|\sin\boldsymbol{\theta}\|_2^2 > 0$, by
\begin{equation}
\sigma_{\max}^2\,\|\sin\boldsymbol{\theta}\|_\infty^2
\;+\; 2\,\sigma_{\max}\,\|\sin\boldsymbol{\theta}\|_\infty\,\frac{\|\mathbf{A}\cos\boldsymbol{\theta} - \mathbf{b}\|_2}{\|\sin\boldsymbol{\theta}\|_2}
\;\le\; \tfrac{1}{2}\,\sigma_{\min}^2.
\label{eq:sufficient}
\end{equation}
Since $|\sin\theta_i| \le 1$ for every $i$, we have $\|\sin\boldsymbol{\theta}\|_\infty \le 1$, and therefore $\|\sin\boldsymbol{\theta}\|_\infty^2 \le \|\sin\boldsymbol{\theta}\|_\infty$. Substituting this into the first term of \eqref{eq:sufficient} shows that \eqref{eq:sufficient} is implied by the linear inequality
\[
\sigma_{\max}\!\left(\sigma_{\max} + 2\,\frac{\|\mathbf{A}\cos\boldsymbol{\theta} - \mathbf{b}\|_2}{\|\sin\boldsymbol{\theta}\|_2}\right)\!\|\sin\boldsymbol{\theta}\|_\infty
\;\le\; \tfrac{1}{2}\,\sigma_{\min}^2,
\]
which rearranges to $\|\sin\boldsymbol{\theta}\|_\infty \le \eta(\boldsymbol{\theta})$.

\textit{(ii)} Let $\boldsymbol{\theta}^* \in B$ be a global minimizer of $L$ on $\mathbb{R}^N$.
If $\sin\boldsymbol{\theta}^* = 0$, the conclusion holds.
Otherwise $\sin\boldsymbol{\theta}^* \ne 0$ and $\|\sin\boldsymbol{\theta}^*\|_\infty \le \eta(\boldsymbol{\theta}^*)$, so (i) applies and yields
\[
L(\tilde{\boldsymbol{\theta}}^*) \;\le\; L(\boldsymbol{\theta}^*) - \tfrac{1}{2}\sigma_{\min}^2\,\|\sin\boldsymbol{\theta}^*\|_2^2 \;<\; L(\boldsymbol{\theta}^*),
\]
contradicting global minimality of $\boldsymbol{\theta}^*$. Hence $\sin\boldsymbol{\theta}^* = 0$.
\end{proof}

\subsubsection{Role of the diagonal shift}
\label{sec:diagonal-shift}

Section~\ref{ising} introduces a tunable diagonal shift
$\mathbf{A}^{T}\mathbf{A} \mapsto \mathbf{A}^{T}\mathbf{A} + \beta \mathbf{I}$
(equivalently $\mathbf{J} \mapsto \mathbf{J} + \beta \mathbf{I}$).
Theorem~\ref{thm:local-binarization} explains the role of this shift precisely,
once the modification is mapped back onto the loss
$L(\boldsymbol{\theta}) = L_{\mathrm{data}}(\boldsymbol{\theta}) + R(\boldsymbol{\theta})$
analyzed above.

The cleanest reading of the shift within the $\cos\boldsymbol{\theta}$ relaxation
is that it augments the regularizer by an extra quadratic penalty on
$\sin\boldsymbol{\theta}$:
\[
R_\beta(\boldsymbol{\theta}) \;=\; \|\mathbf{A}\sin\boldsymbol{\theta}\|_2^2 + \beta\,\|\sin\boldsymbol{\theta}\|_2^2,
\qquad
L_\beta(\boldsymbol{\theta}) \;=\; L_{\mathrm{data}}(\boldsymbol{\theta}) + R_\beta(\boldsymbol{\theta}).
\]
For $\beta > -\sigma_{\min}^2$, the lower bound used in the proof of
Theorem~\ref{thm:local-binarization} sharpens to
\[
R_\beta(\boldsymbol{\theta}) \;\ge\; (\sigma_{\min}^2 + \beta)\,\|\sin\boldsymbol{\theta}\|_2^2,
\]
while the data-term increment estimate (Lemma~\ref{lem:data-increment}) is
unaffected, since it depends only on $\sigma_{\max}(\mathbf{A})$. Repeating the
argument of Theorem~\ref{thm:local-binarization} verbatim with this stronger
regularizer bound yields the modified threshold
\begin{equation}
\eta_\beta(\boldsymbol{\theta}) \;=\;
\frac{\sigma_{\min}^2 + \beta}{2\,\sigma_{\max}\!\left(\sigma_{\max} + 2\,\|\mathbf{A}\cos\boldsymbol{\theta} - \mathbf{b}\|_2 \big/ \|\sin\boldsymbol{\theta}\|_2\right)}
\label{eq:threshold-shifted}
\end{equation}
and the corresponding descent guarantee, valid for every $\boldsymbol{\theta}$
with $\sin\boldsymbol{\theta} \ne 0$ and
$\|\sin\boldsymbol{\theta}\|_\infty \le \eta_\beta(\boldsymbol{\theta})$,
\[
L_\beta(\tilde{\boldsymbol{\theta}}) \;\le\; L_\beta(\boldsymbol{\theta}) - \tfrac{1}{2}\,(\sigma_{\min}^2 + \beta)\,\|\sin\boldsymbol{\theta}\|_2^2.
\]
For $\beta = 0$ this recovers Theorem~\ref{thm:local-binarization}.

In this formulation, $\eta_\beta(\boldsymbol{\theta})$ is monotonically increasing in
$\beta$ on the admissible range $\beta > -\sigma_{\min}^2$ for every fixed
$\boldsymbol{\theta}$ with $\sin\boldsymbol{\theta} \ne 0$: the numerator depends linearly on
$\beta$ while the denominator is independent of $\beta$. Consequently the basin
$B_\beta = \{\boldsymbol{\theta} : \|\sin\boldsymbol{\theta}\|_\infty \le \eta_\beta(\boldsymbol{\theta})\}$
satisfies $B_\beta \supseteq B_0$ for all $\beta \ge 0$, with strict containment
whenever $\beta > 0$. As $\beta$ grows, $\eta_\beta(\boldsymbol{\theta})$ eventually
exceeds the geometric ceiling $\|\sin\boldsymbol{\theta}\|_\infty \le 1$, at which point
the basin saturates to $B_\beta = \mathbb{R}^N$ and the descent inequality
applies globally; the limit $\beta \to \infty$ therefore expresses the regime
in which the regularizer dominates and the binarization mechanism is no longer
restricted to a neighborhood of the binary set.

Two consequences are worth highlighting. First, the diagonal shift is not an
arbitrary regularization knob: it acts as a direct additive lever on the
spectral lower bound $\sigma_{\min}^2$ that controls the binarization
mechanism, and its effect on the basin can be read off from
\eqref{eq:threshold-shifted} without further assumptions on $\mathbf{A}$. Second,
the same calculation clarifies the role of negative $\beta$. Within the
admissible range $-\sigma_{\min}^2 < \beta < 0$ the basin shrinks and the
descent slope $\sigma_{\min}^2 + \beta$ degrades, so
Theorem~\ref{thm:local-binarization} no longer guarantees binarization on an
enlarged set. Negative $\beta$ may nevertheless be useful in practice for
reasons outside the scope of Theorem~\ref{thm:local-binarization}, for
example, destabilizing incorrect spin configurations when the optimizer
stagnates, as observed empirically in Section~\ref{ising}.

\subsection{Ising Model Interpretation}
\label{ising}

To formulate the problem as a physical Ising model, we aim to encourage the continuous ``soft'' variables to converge precisely to the desired discrete states of $\{\pm1\}$, thereby improving the accuracy of the overall optimization process. 

Let $\mathbf{x} \in \{0,1\}^N$ denote a binary state vector. We define the interaction matrix $\mathbf{J}=\mathbf{A}^T \mathbf{A}$ and the external magnetic field $\mathbf{h}=\mathbf{A}^T\mathbf{b}$. The Hamiltonian for this system is given by:
\begin{equation}
    \boldsymbol{\mathcal{H}}  = \mathbf{x}^T \mathbf{J} \mathbf{x} - 2\mathrm{Real}\{ \mathbf{h}^T \mathbf{x} \} \ .
\end{equation}

The implicit regularization $\mathcal{R}$ derived in the previous subsection, can be interpreted mathematically as adding an identity matrix to the interaction matrix (effectively a diagonal shift with a coefficient of 1). We generalize this idea by introducing a tunable hyperparameter $\beta$ that scales this diagonal shift relative to the rest of the matrix. The regularized Hamiltonian becomes:
\begin{equation}
    \boldsymbol{\mathcal{H}} = \mathbf{x}^T \left( \mathbf{J} + \beta\mathbf{I}\right) \mathbf{x} - 2\mathrm{Real}\{ \mathbf{h}^T \mathbf{x} \} \ ,
\end{equation}
where $\beta$ can be dynamically tuned during the optimization process. Note that for $\beta=1$ the result coincides with the complex case.

To transition to a standard physical model, we map the binary variables to bipolar physical spins, defined as $\mathbf{s} = 2\mathbf{x} - 1 \in \{\pm 1\}^N$. Substituting $\mathbf{x} = \frac{1}{2}(\mathbf{s} + 1)$ into the equation yields the final Ising Hamiltonian:
\begin{equation}
\label{eq:ising_hamiltonian0}
    \boldsymbol{\mathcal{H}} = \frac{1}{4}\left({\mathbf{s}+1}\right)^T \left( \mathbf{J} + \beta\mathbf{I}\right) \left({\mathbf{s}+1}\right) - \mathrm{Real}\{ \mathbf{h}^T \left({\mathbf{s}+1}\right) \} \ .
\end{equation}

In the context of optimizing the model, the parameter $\beta$ provides explicit control over the magnitude of the spins. In general, a more positive $\beta$ introduces an additional penalty for spins with higher magnitude, and therefore forces the soft spins to converge closer to zero. However, a more negative $\beta$ forces the soft spins to converge more tightly to $\pm1$ by the end of the optimization. Interestingly, we empirically observe that in certain scenarios, a correctly selected value of $\beta$ helps avoid local minima by destabilizing incorrect spin assignments when the optimization process stagnates.

To determine the optimal regularizer coefficient $\beta$ and evaluate model robustness, we conducted a sensitivity analysis across a symmetric range of real values around zero. While the exact optimal $\beta$ is inherently problem - dependent and distinctly non-zero in most cases, the optimization framework exhibits a broad range of stability.

\subsection{Optimization Algorithm Description}
\label{sub:Algorithm_Description_QUBO}

Across all evaluated optimization problems, the physical spins serve as the optimized parameters, which are continuously updated via gradient descent to minimize the defined energy function. The specific parameterization of these continuous variables depends on the chosen algebraic model: the real and the complex representation are parameterized by a single phase angle, the spherical model by two angles, and the quaternion model by three independent angles. Before the optimization all angle parameters of the model are initialized to random values in the range $[0,2\pi]$.

To ensure robust statistical evaluation, we define two parameters within our framework: \textit{experiments} and \textit{trials}. An \textit{experiment} ($N_{\text{exp}}$) represents a complete optimization task over a newly generated random problem realization, allowing for cross-instance averaging and statistical analysis. Within each experiment, we execute multiple parallel \textit{trials} ($N_{\text{IC}}$), where each trial represents an optimization run starting from a distinct, randomly generated initial condition. The final solution for a given experiment is selected from the trial that achieves the minimum energy.

An exception to this structure is the planted-solution Ising problem. Because these instances are deterministically engineered rather than randomly sampled, only a single experiment ($N_{\text{exp}} = 1$) is given per problem size, and the algorithmic success is evaluated directly against the explicitly known ground-state energy.

The general top-level pseudo-code for this framework is given in Algorithm \ref{alg:training_unified}.

\begin{algorithm}[H]
\caption{General Solver Model Optimization}
\label{alg:training_unified}
\begin{algorithmic}[1] 
\Require
    \Statex $N$ \Comment{Problem size}
    \Statex $N_{\text{exp}}$ \Comment{Number of \textit{experiments}}
    \Statex $N_{\text{IC}}$ \Comment{Number of \textit{trials} (initial conditions)}
\Ensure Average Number of Errors ($\bar{E}_{\text{err}}$)

\For{$n = 1$ \textbf{to} $N_{\text{exp}}$} \Comment{Loop over \textit{experiments}}
    \State Generate random problem realization: $\mathbf{A}, \mathbf{m}, \mathbf{e}$
    \State Calculate measurement vector: $\mathbf{b} = \mathbf{A}\mathbf{m} + \mathbf{e}$
    \State Initialize task-specific hyperparameters (e.g., regularization factor $\beta$)
    
    \For{$k = 1$ \textbf{to} $N_{\text{IC}}$} \Comment{Loop over \textit{trials}}
        \State Initialize optimized parameters $\mathbf{s}$ randomly
        \State Gradient-based minimization of the continuous objective $\mathcal{L}_{\text{optimize}}$
        \State Discretize parameters to valid domain: $\bar{\mathbf{m}} = \text{sgn}(\mathbf{s})$
        \State Evaluate final discrete cost $\mathcal{L}[k] \gets ||\mathbf{A}\bar{\mathbf{m}}-\mathbf{b}||_2$
        \State Calculate number of errors $err[k] = \sum_{i=1}^{N}(\bar{m}_i \neq m_i)$
        \State Update dynamic hyperparameters (if applicable, e.g. regularizers weights)
    \EndFor
    
    \State $k^* \gets \argmin_{k} \mathcal{L}[k]$ \Comment{Select best valid \textit{Trials}}
    \State $\text{res}[n] \gets err[k^*]$ \Comment{Save experiment result}
\EndFor

\State \Return $\bar{E}_{\text{err}} \gets \frac{1}{N_{\text{exp}}} \sum_{n=1}^{N_{\text{exp}}} \text{res}[n]$
\end{algorithmic}
\end{algorithm}

While the core pipeline remains consistent, specific adaptations are required depending on the problem class. For instance, in the sparse coding problem, we construct a cost function that incorporates the cardinality constraint via a dynamic regularization multiplier. Practically, the penalty factor $\lambda$ is initialized to an empirical baseline (e.g., $0.035$) and is iteratively adjusted at the end of each trial: $\lambda \gets \lambda + \alpha \cdot (C[k] - C_{\text{target}})$, where $\alpha = 10^{-3}$ is the update weight and $C[k]$ is the measured cardinality of the discretized vector. This operation is meant to be done in step 11 of Algorithm \ref{alg:training_unified}. 

The overall computational complexity of the algorithm per experiment scales as $\mathcal{O}(N_{\text{IC}} \cdot T)$, where $N_{\text{IC}}$ denotes the number of distinct initial conditions and $T$ represents the number of time steps (gradient updates) required for convergence in the optimization procedure.

\section{Simulation Results}
\label{sec:TestRes}

In this section, we present the simulation results across all evaluated optimization problems. For both the QUBO and sparse coding tasks, our evaluation framework consisted of $N_{\text{exp}} = 50$ experiments, with $N_{\text{IC}} = 20$ trials. The continuous parameters were updated using the Adam optimizer, as it showed fast convergence. For the sparse coding task specifically, we utilized the LightSolver codebase \cite{romano2024quantum} that benchmarked our approach against standard recovery algorithms, namely Orthogonal Matching Pursuit (OMP) \cite{pati1993orthogonal} and the Least Absolute Shrinkage and Selection Operator (LASSO) \cite{tibshirani1996regression}. 

For the sparse coding task, to establish rigorous baselines for sufficiently small-scale instances, we performed an exhaustive combinatorial search to determine the exact optimal solution. Across all tested problem instances, the total optimization runtime varied from several hours up to two days on a standard personal computer, we used an Apple M1 Pro (32 GB RAM).

\subsection{QUBO}
\label{sub:quboResults}
For the QUBO evaluation (subsection \ref{sub:formQUBO}), we examined a system of size $M=N=160$, utilizing target ground-truth binary vectors with a fixed cardinality of $C=80$, evaluated across varying Gaussian noise levels $\sigma$. 

For each algebraic parameterization model at every noise level, we replace the constant diagonal hyperparameter $\beta$ with a two-stage parameters, defined by $\mathbf{k}_{\mathrm{reg}} = (k_0, k_1)$. These hyperparameters govern the regularization strength during the first and second halves of the 2000-epoch optimization process, respectively, with the transition occurring at epoch 1000. The specific values for $k_0$ and $k_1$ were determined via grid search. We introduced this modification because empirical observations demonstrate that this dynamical changing of the strength of the regularizer helps the algorithm escape local minima.

Figure \ref{fig:qubo_sweep} summarizes the performance comparisons. A clear trend emerges from these results: in the absence of explicit regularization, the standard real-valued relaxation performs significantly worse than the multidimensional parameterizations. However, upon introducing the optimized regularization schedule, the performance gap substantially narrows, rendering all models roughly equivalent in accuracy (within statistical fluctuations).

\begin{figure}[H]
\centering
\includegraphics[width=1.0\textwidth]{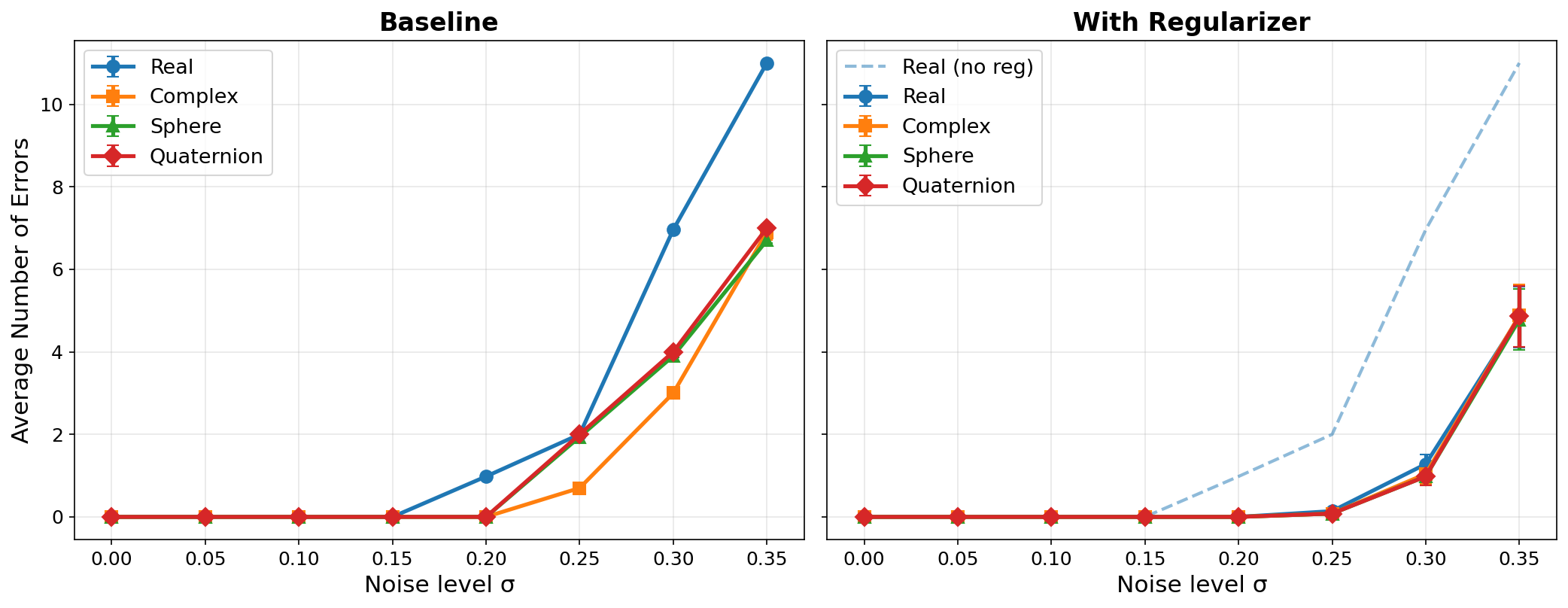}
\caption{Average number of bit errors vs.\ noise level $\sigma$ for the QUBO $160 \times 160$ problem (cardinality $80$, $10^4$ epochs, $50$ experiments $\times$ $20$ trials). Left\textbf{:} without regularization. Right\textbf{:} with optimized two-stage diagonal regularization ($\mathbf{k}_{\mathrm{reg}} = (0,1.0)$). Error bars indicate the standard error of the mean.}
\label{fig:qubo_sweep}
\end{figure}

\subsection{Sparse Coding}
For the sparse coding evaluation (subsection \ref{sub:formSC}), we investigated two distinct problem scales: a small-scale, solvable instance with dimensions $M=8, N=16$, and a target cardinality of $C=6$, and a larger, more complex instance with $M=80, N=160$, and $C=30$. 
The comparative performance of the evaluated models across varying noise levels is depicted in Figure \ref{fig:SparseCodingRes}.

\begin{figure}[H]
    \centering
    \includegraphics[width=1.0\linewidth]{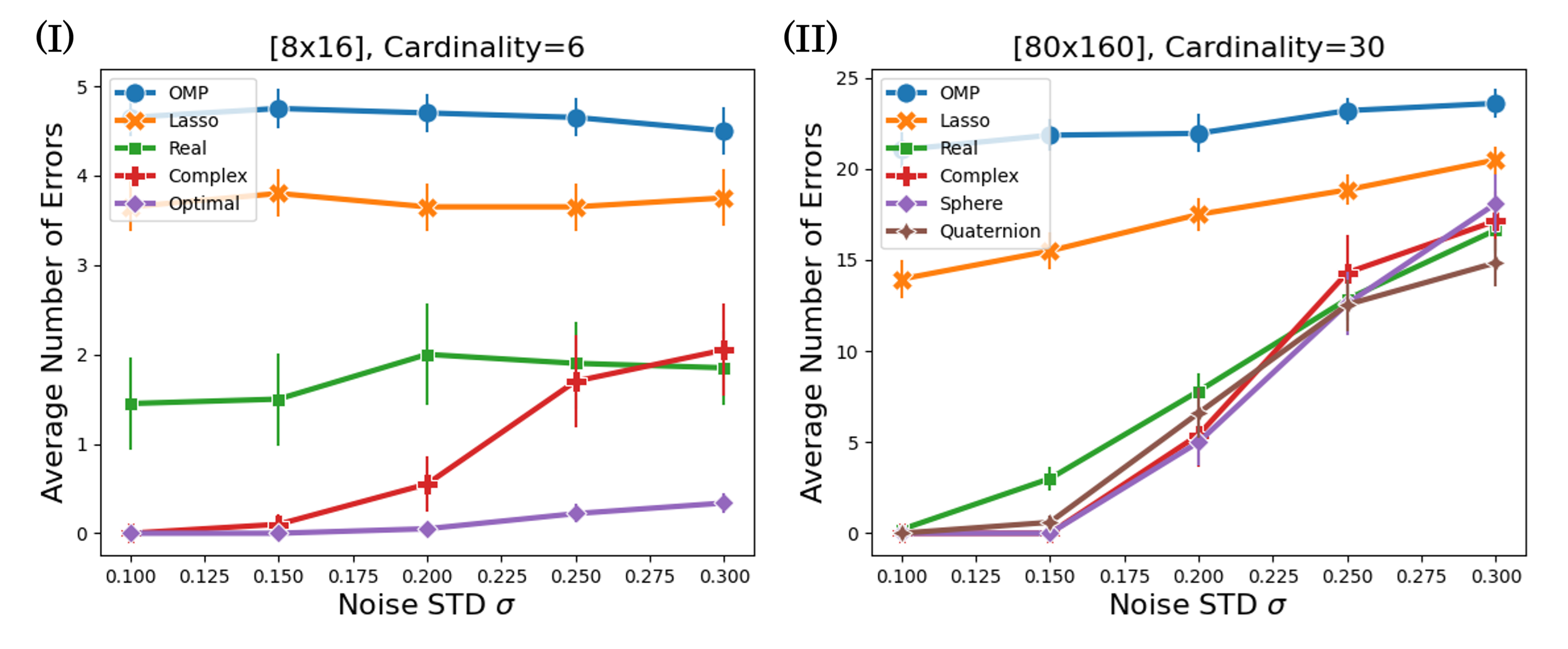}
    \caption{Summary of sparse coding recovery results. Panels (I) and (II) present a comparison of the different optimization methods under the specific conditions indicated in their titles. Each data point on the graph represents the average number of bit errors for a specific model at a given noise power. Averages were computed over 20 distinct random problem realizations, utilizing 20 trials with random initial conditions for each realization.}
    \label{fig:SparseCodingRes}
\end{figure}

\subsection{Ising Problems with Planted Solutions}
\label{ising1}
In this final evaluation phase, we benchmarked the models against 11 specific planted-solution Ising problem instances (subsection \ref{sub:formPS}), with system sizes ranging from $N=28$ to $N=1188$. Unlike the dense matrices in the previous tasks, the interaction matrices $\mathbf{J}$ (defined in Eq. \ref{eq:ising_hamiltonian0}) utilized for these instances are inherently sparse, and asymmetric, and possess zero diagonal elements ($\mathbf{J}_{ii} = 0$). Physically, it means that the spin-spin interactions are non-reciprocal. However, since the quadratic term only use symmetry part, we practically used $\mathbf{J}\rightarrow\tfrac{1}{2}(\mathbf{J}+\mathbf{J}^T)$.

Because each planted solution instance represents a single, deterministically engineered problem with a known, exact ground state, cross-instance averaging is not applicable. We maximized the number of independent search trials ($N_{\text{IC}}$) within a practical computational time constraint of approximately a few hours to a single day per problem. The exact number of trials executed for each problem is detailed in the tables below. For these instances, we focused exclusively on the real parameterization, as higher-order encodings showed no significant advantage over the regularized real-valued model.

Table \ref{tab:resPS} compares the success rates of the baseline model (standard real parametrization) against our proposed model (real parametrization with regularization). Our approach successfully resolved 8 of the 11 problems, demonstrating a clear improvement over the baseline, which solved only 2. Alongside the raw success counts, the table reports the average number of trials required to reach a solution. While performance generally decreases as system size grows, we observe occasional fluctuations that we attribute to the effects of numerical fine-tuning. Further analysis of the convergence dynamics for these planted solutions is detailed in \ref{app:histPlanted}. An example for a planted solution problem is shown in \ref{app:additional_results}.

\begin{table}
    \centering
    \begin{tabular}{|c|c|c|c|c|c|}\hline
         \textbf{\#} &  \textbf{Size}  &\textbf{GS Energy}&  \textbf{Trials ($10^3$)} &  \textbf{Regularized} & \textbf{Baseline} \\\hline
         1&  28  &-70&  1 &  \textcolor{ForestGreen}{ 87 (12)}& {\color{red}9 (111)}\\\hline
         2&  40  &-101&  1 &  \textcolor{ForestGreen}{ 46 (22)}& {\color{red}4 (250)} \\\hline
         3&  66  &-157&  6 &  {\color{ForestGreen}971 (7)}& 0\\\hline
         4&  66  &-164&  6 &  \textcolor{ForestGreen} {14 (429)}& 0\\\hline
         5&  82  &-194&  6 &  \textcolor{ForestGreen} {6 (1000)}& 0\\\hline
         6&  93  &-223&  6 &  {\color{ForestGreen} 103 (59)} & 0\\\hline
         7&  153  &-364&  6 &  {\color{ForestGreen}212 (29)} & 0\\\hline
         8&  336  &-808&  6 &  {\color{ForestGreen}1 (6000)} & 0\\\hline
    \end{tabular}
    \caption{Success rates for finding the planted solution over the different problems, along with the problems size and ground state energies. The values in parentheses indicate the average number of trials to success.}
    \label{tab:resPS}
\end{table}

\section{Discussion and Conclusion}
\label{sec:discussion}

The primary contribution of this work is the identification of a regularization mechanism revealed by complex phase embedding, which can be applied more generally to continuous relaxations of discrete optimization problems. By adopting a physical-inspired approach, we demonstrated the efficacy of this complex relaxation framework for solving NP-hard combinatorial optimization problems. While naively relaxing discrete spin variables into a continuous domain often causes algorithms to stall in suboptimal local minima, our analysis reveals that the initial choice of mathematical parameterization is paramount. Specifically, adopting a physical, complex-wave interpretation for the discrete bits exposed a natural, implicit regularization mechanism. When this physics-inspired regularizer is explicitly incorporated, it drastically improves the performance of standard, real-valued optimization.

Furthermore, our empirical evaluations highlight a distinct plateau in representational advantage. Crucially, our analytical derivations (subsection \ref{sub:complexRelaxation}) and empirical evaluations reveal that the implicit binarization penalty is fully realized within a two-dimensional complex space; extending the parameterization into higher-dimensional structures, such as spherical or quaternion algebras, introduces computational overhead without yielding any additional regularizing benefits over the complex-valued approach.

Additionally, when benchmarked in heavily constrained, under-determined scenarios like sparse coding, our regularized framework exhibited highly competitive performance against traditional algorithms (e.g., LASSO, OMP) and contemporary dedicated platforms like LightSolver, reliably identifying the exact optimal sparse supports.

We posit that the success of this physics-inspired approach stems directly from the mathematical properties of the unit-modulus complex representation. By mapping the continuous variables onto the complex plane, the optimization objective naturally decouples into a standard data fidelity term and a distinct residual term. Minimizing this residual intrinsically acts as a penalty that forces the continuous variables toward the discrete binary boundaries. By extracting this topological insight and adapting it as a tunable regularizer for real-valued spaces, we grant gradient-based optimizers the ability to smoothly navigate non-convex energy landscapes and successfully bypass part of the severe local optima typical of unregularized continuous relaxations.

Despite the successful performance, our approach has several notable limitations. First, the effectiveness of the binarization mechanism depends on the spectral characteristics of the Ising interaction matrix; when the matrix exhibits strong anisotropy or correlations, the binarization penalty no more hold (as discussed in Appendix \ref{app:complexReg}). In addition, identifying an appropriate value for the regularization coefficient $\beta$ is sometimes challenging, and the current framework has difficulty preserving exact convergence when extended to very large-scale problem instances.

From a machine learning perspective, this work motivates a rigorous, systematic comparison between standard real-valued architectures and Complex-Valued Neural Networks (CVNNs) across a broader class of combinatorial optimization problems to see if similar implicit regularizations naturally emerge. Finally, the mathematical affinity between our complex unit-circle parameterization and the phase of quantum amplitudes suggests a compelling link between our classical unit-circle parameterization and quantum phase amplitudes, which warrants further investigation. Investigating how these classical physics-inspired solvers might inform, or interface with, near-term quantum computing algorithms remains a compelling avenue for our future research.

\section*{Acknowledgments}
The authors thank Itay Hen for sharing the planted solution problems that were investigated during the course of this project. K.C. gratefully acknowledges the Milner Foundation, and the VATAT (PBC) Fellowship for Outstanding PhD Students in Data Science. This work was supported by the Israeli Ministry of Innovation, Science and Technology (Grant No. 1001572598); the Tel Aviv University Center for Artificial Intelligence; the Israel Science Foundation (ISF) (Excellence Center Grant No. 2312/21 and Grant No. 969/22). Y.O. is supported by the US-Israel Binational Science Foundation (BSF), the Israel Ministry of science, the Tel Aviv University Center for Artificial Intelligence, and the PAZY Foundation.



\section*{Code and Data Availability}
All underlying data, configurations, and custom Python scripts required to fully reproduce the empirical results of this study are publicly available at \href{https://github.com/khencohen/ComplexIsing}{GitHub Repo}.

\bibliographystyle{elsarticle-num} 
\bibliography{cas-refs}

\newpage

\appendix

\section{Approximate Isometry and Binarization}
\label{app:complexReg}

In this appendix we develop a random-matrix complement to
Theorem~\ref{thm:local-binarization} that explains why its spectral
hypothesis $\sigma_{\min}(\mathbf{A}) > 0$ is generic for typical measurement
ensembles, and provides intuition for why the residual term
$\mathcal{R} = \|\mathbf{A}\sin(\boldsymbol{\theta})\|_2^2$ acts as a
binarization-promoting penalty. The deterministic, quantified version of the
binarization claim is the content of Theorem~\ref{thm:local-binarization};
the calculation here is intended as motivation, not as a substitute for it.

To avoid clashing with the noise standard deviation $\sigma$ of
Section~\ref{sec:problemsformulation} and the singular values $\sigma_{\min}$,
$\sigma_{\max}$ of Section~\ref{sec:binarization}, we denote the
per-entry variance of the random measurement matrix by $\nu^2$.

Let $\mathbf{A} \in \mathbb{R}^{M \times N}$ be drawn from a sub-Gaussian
ensemble (in the sense of \cite{vershynin2018high}) with i.i.d. entries
$A_{ij}$ of zero mean and variance $\nu^2$. For any fixed vector
$\mathbf{v} \in \mathbb{R}^N$,
\begin{equation}
\mathbb{E}\!\left[\|\mathbf{A}\mathbf{v}\|_2^2\right]
\;=\; \mathbb{E}[\mathbf{v}^T \mathbf{A}^T \mathbf{A}\,\mathbf{v}]
\;=\; \mathbf{v}^T\,\mathbb{E}[\mathbf{A}^T\mathbf{A}]\,\mathbf{v}.
\end{equation}
The Gram matrix $\mathbf{G} = \mathbf{A}^T\mathbf{A}$ has diagonal entries with expected
value $M\nu^2$ and off-diagonal entries with vanishing expectation (by
independence and the zero-mean assumption), so
$\mathbb{E}[\mathbf{A}^T\mathbf{A}] = \alpha\,\mathbf{I}_N$ with $\alpha \equiv M\nu^2$, and the
operation acts as a scaled isometry in expectation:
\begin{equation}
\mathbb{E}\!\left[\|\mathbf{A}\mathbf{v}\|_2^2\right]
\;=\; \alpha\,\|\mathbf{v}\|_2^2.
\end{equation}

For sub-Gaussian ensembles, the Hanson-Wright inequality (Theorem~6.2.1
in~\cite{vershynin2018high}) gives a quantitative two-sided concentration
bound: there is an absolute constant $c > 0$ such that, for every fixed
$\mathbf{v} \in \mathbb{R}^N$ and every $\varepsilon \in (0,1)$,
\begin{equation}
\label{eq:isometry-fixed-v}
\mathbb{P}\!\left(
\bigl|\,\|\mathbf{A}\mathbf{v}\|_2^2 - \alpha\,\|\mathbf{v}\|_2^2\,\bigr|
\;\ge\; \varepsilon\,\alpha\,\|\mathbf{v}\|_2^2
\right)
\;\le\; 2\,\exp\!\bigl(-c\,M\,\varepsilon^2\bigr).
\end{equation}
A standard $\varepsilon$-net argument (e.g., Theorem~9.1.1 in
\cite{vershynin2018high}) extends \eqref{eq:isometry-fixed-v} from a fixed
$\mathbf{v}$ to a uniform statement: provided $M \gtrsim N$ up to logarithmic
factors, with high probability
\begin{equation}
(1-\varepsilon)\,\alpha\,\|\mathbf{v}\|_2^2
\;\le\; \|\mathbf{A}\mathbf{v}\|_2^2
\;\le\; (1+\varepsilon)\,\alpha\,\|\mathbf{v}\|_2^2
\qquad\text{for all } \mathbf{v} \in \mathbb{R}^N.
\label{eq:isometry-uniform}
\end{equation}
This uniformity is what is needed to apply the estimate to the
\emph{trajectory} $\sin(\boldsymbol{\theta}(t))$ generated by the optimizer.
Specializing to $\mathbf{v} = \sin(\boldsymbol{\theta})$ gives, for every iterate
along the trajectory,
\begin{equation}
\|\mathbf{A}\sin(\boldsymbol{\theta})\|_2^2 \;\approx\; \alpha\,\|\sin(\boldsymbol{\theta})\|_2^2,
\label{eq:R-approx}
\end{equation}
with relative error at most $\varepsilon$.

Combined with the deterministic lower bound
$R(\boldsymbol{\theta}) \ge \sigma_{\min}^2\,\|\sin(\boldsymbol{\theta})\|_2^2$
used in Theorem~\ref{thm:local-binarization}, the estimate
\eqref{eq:R-approx} pins down a typical value
$\sigma_{\min}^2 \gtrsim (1-\varepsilon)\,\alpha = (1-\varepsilon)\,M\nu^2$
for sub-Gaussian ensembles with $M \gtrsim N$, ensuring that the basin
threshold $\eta(\boldsymbol{\theta})$ in Theorem~\ref{thm:local-binarization}
is bounded away from zero for typical instances. The actual binarization
statement, that any iterate sufficiently close to a binary configuration
is improved by projecting onto it, is the content of
Theorem~\ref{thm:local-binarization} itself, not of the heuristic identity
\eqref{eq:R-approx}.

\paragraph{Underdetermined regime}
The argument above relies critically on $M \gtrsim N$. In the sparse-coding
setting, where $M \ll N$, the matrix
$\mathbf{A}^T\mathbf{A}$ has a $(N-M)$-dimensional null space, $\sigma_{\min}(\mathbf{A}) = 0$,
and Theorem~\ref{thm:local-binarization} is vacuous: the regularizer
$\mathcal{R}$ penalizes $\sin$-components in the row space of $\mathbf{A}$ but is
blind to the null space. In this regime, the explicit cardinality
constraint enforced through the dynamic Lagrange multiplier restricts the feasible set sufficiently for the
relaxed solver to recover binary supports. The empirical performance reported therefore reflects the combined effect of
the regularizer and the cardinality constraint, rather than the
isometry-based mechanism considered in this appendix on its own.

More generally, if $\mathbf{A}$ exhibits strong column correlations or highly
anisotropic singular values without being explicitly low-rank, the
conclusion of \eqref{eq:R-approx} weakens: directions associated with small
singular values of $\mathbf{A}$ are penalized only weakly, and binarization is not
enforced uniformly across coordinates. This does not invalidate the
underlying mechanism, but it makes its strength track the conditioning and
spectral structure of the measurement matrix.

\subsection{Regularization Term in Higher-Dimensional Parametrizations}
\label{app:hyperspherical}

We now show that replacing the unit-circle parametrization
$\mathbf{x} = \cos(\boldsymbol{\theta})$ by a unit-sphere parametrization on
$S^{d-1}$ with $d \ge 2$ does not yield any additional regularization beyond
the $d = 2$ case.

Parametrize each spin as a unit vector
$\mathbf{u}_n = (u_n^{(1)}, u_n^{(2)}, \ldots, u_n^{(d)}) \in S^{d-1}$,
with the first component $u_n^{(1)}$ playing the role of the binary
coordinate (eventually rounded to $\pm 1$) and the remaining $d-1$
components representing the relaxation degrees of freedom. Let
$\mathbf{u}^{(i)} \in \mathbb{R}^N$ collect the $i$-th components across
spins, $(\mathbf{u}^{(i)})_n = u_n^{(i)}$. The data fidelity uses only
$\mathbf{u}^{(1)}$, and the regularizer extends naturally as
\[
\mathcal{R}_{d\ge 2}(\mathbf{u}) \;=\; \sum_{i=2}^{d} \|\mathbf{A}\,\mathbf{u}^{(i)}\|_2^2.
\]
Applying the isometry estimate \eqref{eq:R-approx} componentwise,
\begin{equation}
\mathcal{R}_{d\ge 2}(\mathbf{u})
\;\approx\; \alpha\,\sum_{i=2}^{d}\|\mathbf{u}^{(i)}\|_2^2
\;=\; \alpha\,\sum_{n=1}^{N}\sum_{i=2}^{d}\bigl(u_n^{(i)}\bigr)^2.
\end{equation}
The unit-norm constraint $\sum_{i=1}^{d}(u_n^{(i)})^2 = 1$ for each $n$ gives
\begin{equation}
\mathcal{R}_{d\ge 2}(\mathbf{u})
\;\approx\; \alpha\,\sum_{n=1}^{N}\Bigl(1 - (u_n^{(1)})^2\Bigr).
\label{eq:R-d-cartesian}
\end{equation}
The right-hand side depends only on the binary coordinate $u_n^{(1)}$ and is
independent of $d$. For $d=2$ with the parametrization
$(u_n^{(1)}, u_n^{(2)}) = (\cos\theta_n, \sin\theta_n)$, the expression
\eqref{eq:R-d-cartesian} reduces to
$\alpha\,\|\sin\boldsymbol{\theta}\|_2^2$. For $d > 2$, the same expression is
obtained, so increasing the dimension provides no additional regularization
benefit beyond what the two-dimensional complex parametrization already
realizes.

\if{
\section{Approximate Isometry and Binarization}
\label{app:complexReg}
In this appendix, we demonstrate that the secondary residual term $\mathcal{R} = \|\mathbf{A}\sin(\boldsymbol{\theta})\|_2^2$ acts as an implicit binarization regularizer. 

We approach this by analyzing the expected behavior of a random matrix ensemble. While a single instantiation of the measurement matrix $\mathbf{A}$ operates deterministically, large random matrices exhibit strong self-averaging properties due to the concentration of measure in high-dimensional spaces. Specifically, let us assume $\mathbf{A} \in \mathbb{R}^{M \times N}$ is drawn from a sub-Gaussian ensemble (for the precise definition of sub-Gaussian ensemble we refer to \cite{vershynin2018high}) whose entries $A_{ij}$ are independent and identically distributed (i.i.d.) with zero mean and variance $\sigma_A^2$. For any fixed vector $\mathbf{v} \in \mathbb{R}^N$, the expected squared Euclidean norm of the mapped vector $\mathbf{A}\mathbf{v}$ over the matrix ensemble is given by:
\begin{equation}
    \mathbb{E}[\|\mathbf{A} \mathbf{v}\|_2^2] = \mathbb{E}[\mathbf{v}^T \mathbf{A}^T \mathbf{A} \mathbf{v}] = \mathbf{v}^T \mathbb{E}[\mathbf{A}^T \mathbf{A}] \mathbf{v} \ .
\end{equation}

By evaluating the $N \times N$ Gram matrix $\mathbf{G} = \mathbf{A}^T \mathbf{A}$, we find that the diagonal elements have an expected value of $M\sigma_A^2$, while the off-diagonal elements vanish due to the statistical independence and zero mean of the entries. Consequently, $\mathbb{E}[\mathbf{A}^T \mathbf{A}] = \alpha \mathbf{I}_N$, where we define the scaling factor $\alpha \equiv M\sigma_A^2$. Thus, the operation acts as a scaled isometry in expectation:
\begin{equation}
    \mathbb{E}[\|\mathbf{A} \mathbf{v}\|_2^2]  = \alpha \|\mathbf{v}\|_2^2 \ .
\end{equation}

For sub-Gaussian matrices, standard concentration inequalities guarantee that the variance of $\|\mathbf{A}\mathbf{v}\|_2^2$ is heavily suppressed \cite{vershynin2018high}. Therefore, in the high-dimensional limit, this proportionality is preserved with high probability for a typical single realization:
\begin{equation}
    \|\mathbf{A}\mathbf{v}\|_2^2 \approx \alpha\|\mathbf{v}\|_2^2 \ .
\end{equation}

Applying this result to the imaginary component of our complex relaxation, where $\mathbf{v} = \sin(\boldsymbol{\theta})$, yields:
\begin{equation}
    \|\mathbf{A}\sin(\boldsymbol{\theta})\|_2^2 \approx \alpha \|\sin(\boldsymbol{\theta})\|_2^2 \ .
\end{equation}

Because the global optimization framework actively minimizes this residual term as part of the total cost function, the optimizer naturally drives the norm $\|\sin(\boldsymbol{\theta})\|_2^2 \to 0$. This forces the continuous phase variables to converge to $\theta_j = \pi k$ for $k \in \mathbb{Z}$, which inherently projects the relaxed variables back to the discrete binary field: $\cos(\theta_j) = \pm 1$.

However, it is important to acknowledge the limitations of this mechanism when the assumptions of the random matrix ensemble are relaxed. If the matrix $\mathbf{A}$ exhibits strong column correlations, highly anisotropic singular values, or a low-rank structure, the scaled isometry property breaks down. In such cases, certain non-binary directions corresponding to small singular values may be only weakly penalized, or could even lie close to the null space of $\mathbf{A}$. Consequently, the regularizer does not enforce binarization equally across all coordinates. While this does not invalidate the underlying mechanism, it highlights that its strength and effectiveness are fundamentally dependent on the spectral and structural properties of the measurement matrix.

\subsection{Regularization Term in Hyperspherical Coordinates}

Assuming this case hold for any higher dimension, and denote the angles as $\{ \theta^{(i)} \}_{i=1}^d$. Therefore, one can notice that for the general case:
\begin{gather}
    \mathcal{R}_{\text{d}\geq 2}(\boldsymbol{\theta}) = \alpha \sum_{i=2}^d \|\sin(\boldsymbol{\theta}^{(i)})\|_2^2 = \alpha \sum_{i=2}^d \sum_{n=1}^N \sin^2(\theta_n^{(i)}) \\
    = \alpha \sum_{n=1}^N \sum_{i=2}^d \sin^2(\theta_n^{(i)}) =
    \alpha \sum_{n=1}^N \left( 1-\sin^2(\theta_n^{(1)}) \right) ,
\end{gather}
while in the second equation we used the definition of the l2 norm, in the third we exchanged the summations, and in the third step we used trigonometric identity using the fact that the vector is normalized.

Consequently, we find that $\mathcal{R}_{d\geq2}(\boldsymbol{\theta})$ depends solely on $\theta_1$. This implies that increasing the dimension $d$ provides no advantage for this regularization objective compared to the two-dimensional case.

}\fi
\section{Lasso Regularization}
\label{app:LassoReg}

To fully replicate the objective outlined in Eq. \ref{eq:quboLoss}, we must explicitly reintroduce the sparsity constraint using a continuous penalty.  Substituting the complex relaxation (please see $l1$ derivation) fundamentally alters the nature of this sparsity regularizer. The ultimate comparison between the real and complex optimization landscapes is given by:
\begin{align} \label{eq:finalLossComparison}
    \mathcal{L}_{\text{Real}} &= \min_{\boldsymbol{\theta} \in \mathbb{R}^N} \left( \underbrace{\|\mathbf{A}\cos(\boldsymbol{\theta}) - \mathbf{b} \|_2^2}_{\mathcal{L}_{\text{data}}} + \lambda \underbrace{\| \cos(\boldsymbol{\theta}/2)\|_2^2}_{\mathcal{L}_{\text{reg}}} \right) \ , \\
    \mathcal{L}_{\text{Complex}} &= \min_{\boldsymbol{\theta} \in \mathbb{R}^N} \left( \underbrace{\|\mathbf{A}\cos(\boldsymbol{\theta}) - \mathbf{b} \|_2^2}_{\mathcal{L}_{\text{data}}} + \underbrace{\|\mathbf{A}\sin(\boldsymbol{\theta})\|_2^2}_{\mathcal{R}} + \lambda \underbrace{ \| \cos({\boldsymbol{\theta}/2})\|_1 }_{\mathcal{L}_{\text{Lasso}}} \right) \ ,
\end{align}
where $\lambda$ acts as the hyperparameter controlling the target cardinality. Notably, the complex formulation naturally induces an $\ell_1$-type Lasso penalty on the relaxation angles, yielding fundamentally different and more stable convergence properties than the $\ell_2$ penalty present in the real-valued baseline.

\subsection{$l1$ Derivation}
We start with equation \ref{eq:complexLoss} as presented in the paper. Then, adding \textit{l1} regularization term, the \textbf{real} optimization is defined as:
\begin{gather}
    \min_{\boldsymbol{\theta} \in \mathbb{R}^N} \|\mathbf{A}\cos(\boldsymbol{\theta}) - \mathbf{b} \|_2^2 + \frac{\lambda}{2} \|\cos(\boldsymbol{\theta}) + 1 \|_1  = \\
    \min_{\boldsymbol{\theta} \in \mathbb{R}^N} \underbrace{\|\mathbf{A}\cos(\boldsymbol{\theta}) - \mathbf{b} \|_2^2}_{\mathcal{L}_{\text{real}}} + \lambda \underbrace{\| \cos(\boldsymbol{\theta}/2)\|_2^2}_{\mathcal{L}_{\text{reg}}}
\end{gather}
While we used the fact that: 
\begin{gather}
    \|\cos(\boldsymbol{\theta}) + 1 \|_1=\sum_n |{ \cos(\theta_n) + 1}|=2\sum_n \cos^2(\theta_n/2)= 2\| \cos(\boldsymbol{\theta}/2)\|_2^2
\end{gather}

And for the \textbf{complex} optimization we get:
\begin{align} \label{eq:complexLoss}
   & \min_{\boldsymbol{\theta} \in \mathbb{R}^N} \|\mathbf{A} e^{i\boldsymbol{\theta}} - \mathbf{b} \|_2^2 + \frac{\lambda}{2} \|e^{i\boldsymbol{\theta}} + 1 \|_1 = \nonumber \\
    & \min_{\boldsymbol{\theta} \in \mathbb{R}^N} \underbrace{\|\mathbf{A}\cos(\boldsymbol{\theta}) - \mathbf{b} \|_2^2}_{\mathcal{L}_{\text{real}}} + \underbrace{\|\mathbf{A}\sin(\boldsymbol{\theta})\|_2^2}_{\mathcal{R}} +  \lambda \underbrace{ \| \cos({\boldsymbol{\theta}/2})\|_1 }_{\mathcal{L}_{\text{Lasso}}} \ .
\end{align}

While we used the fact that:
\begin{gather}
    \|e^{i\boldsymbol{\theta}} + 1 \|_1= \sum_n |e^{i{\theta_n}} + 1|= \\ \nonumber \sum_n \sqrt{ |e^{i{\theta_n}} + 1|^2}=\sum_n \sqrt{ (\cos({\theta_n}) + 1)^2+\sin^2(\theta_n)} \\ \nonumber
    = \sum_n \sqrt{ \cos^2({\theta_n})+\sin^2(\theta_n)+1+2\cos({\theta_n})} = \\ \nonumber
     \sqrt{2}  \sum_n \sqrt{ 1+\cos({\theta_n})} = 2 \sum_n \sqrt{ \cos^2({\theta_n/2})} =\\ \nonumber 2 \sum_n | \cos({\theta_n/2})|
     =  2 \| \cos({\boldsymbol{\theta}/2})\|_1
\end{gather}
We use these relations in eq. \ref{eq:optimization} for the derivation.

\newpage
\section{Histograms of Results - Planted Solutions}
\label{app:histPlanted}

This appendix presents the distributions of the estimated ground state energies across the tested problem instances using real relaxation model with a regularizer. In the provided histograms, the horizontal axis denotes the estimated ground state energy obtained at the end of the optimization process for each trial, while the vertical axis indicates the number of trials which resulted in each energy value.

\begin{figure}[H]
\centering
\includegraphics[width=0.32\linewidth]{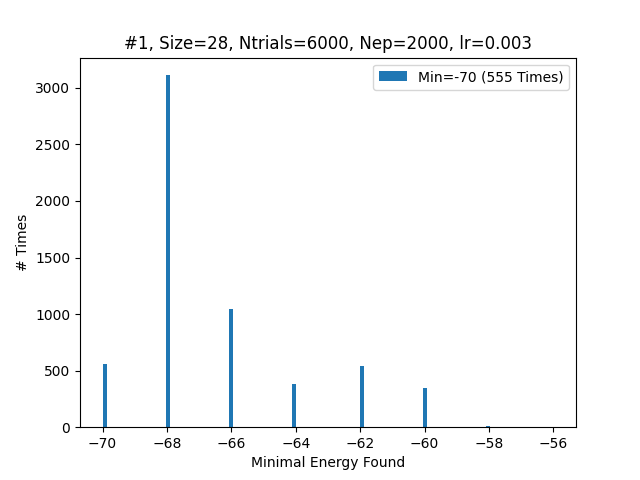}
\includegraphics[width=0.32\linewidth]{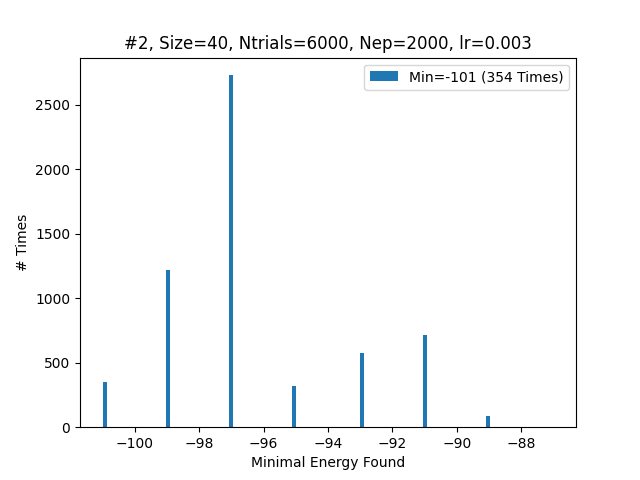}
\includegraphics[width=0.32\linewidth]{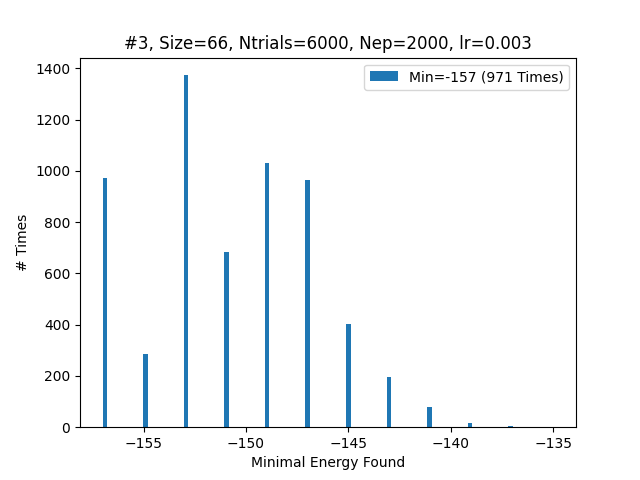}
\includegraphics[width=0.32\linewidth]{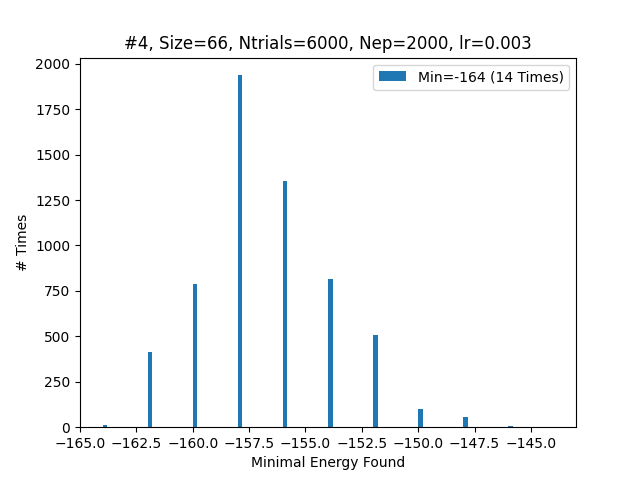}
\includegraphics[width=0.32\linewidth]{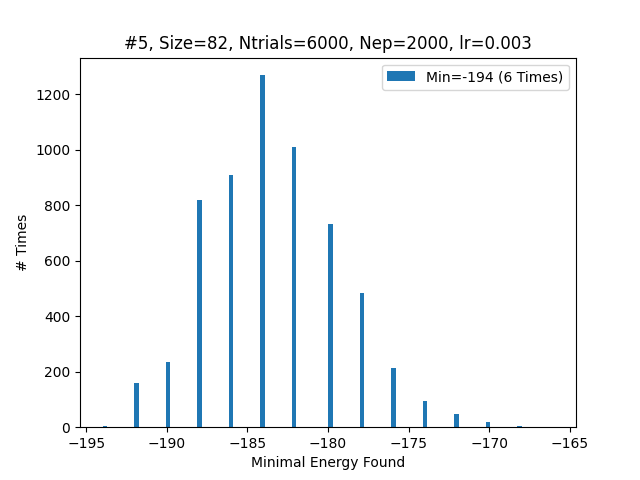}
\includegraphics[width=0.32\linewidth]{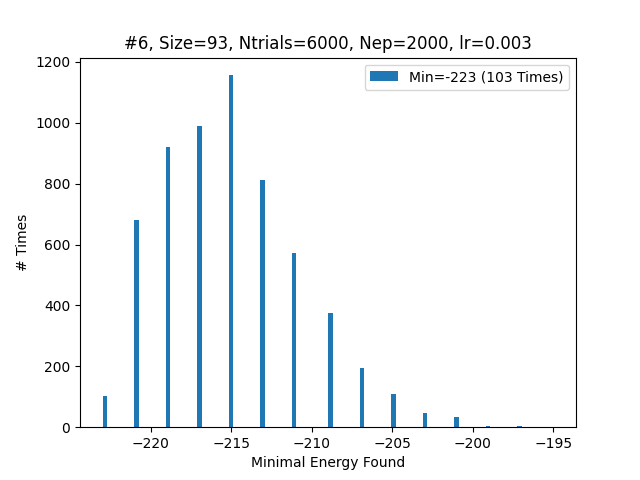}
\includegraphics[width=0.32\linewidth]{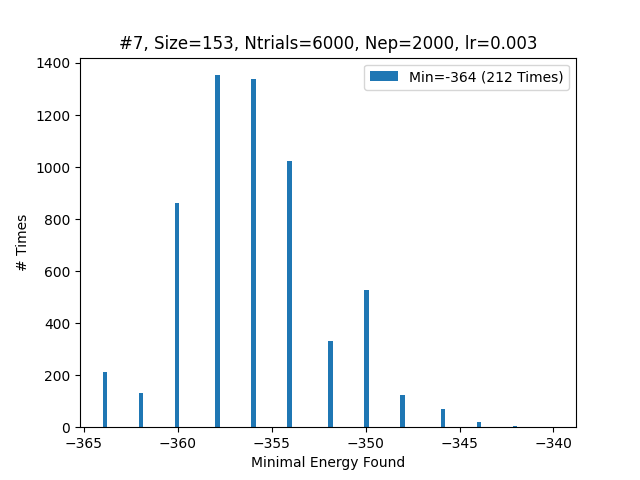}
\includegraphics[width=0.32\linewidth]{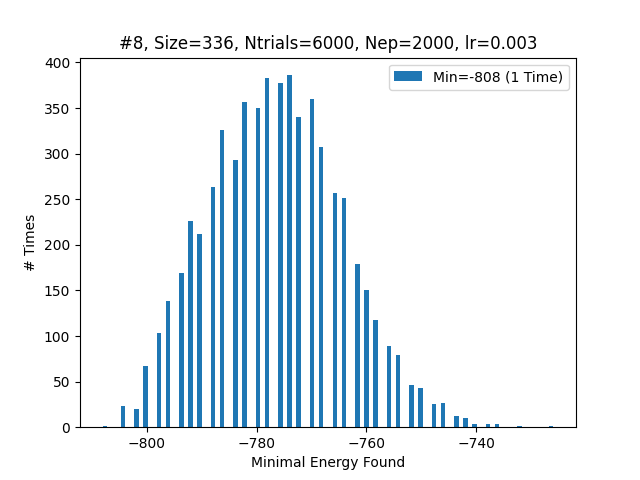}
\caption{Problems 1-8 - Histograms of energy distributions of the results for different trials.}
\label{pq_hist}
\end{figure}




\newpage
\section{Planted Solution Problem Example}
\label{app:additional_results}

Figure \ref{fig:J0} provides a graphical representation of the Hamiltonian for the smallest problem instance (Problem 1). In this visualization, the leftmost column depicts the local magnetic field vector $\mathbf{h}$, while the subsequent columns represent the elements of the spin-spin coupling matrix $\mathbf{J}$.
\begin{figure}[H]
\centering
\includegraphics[width=1.0\linewidth]{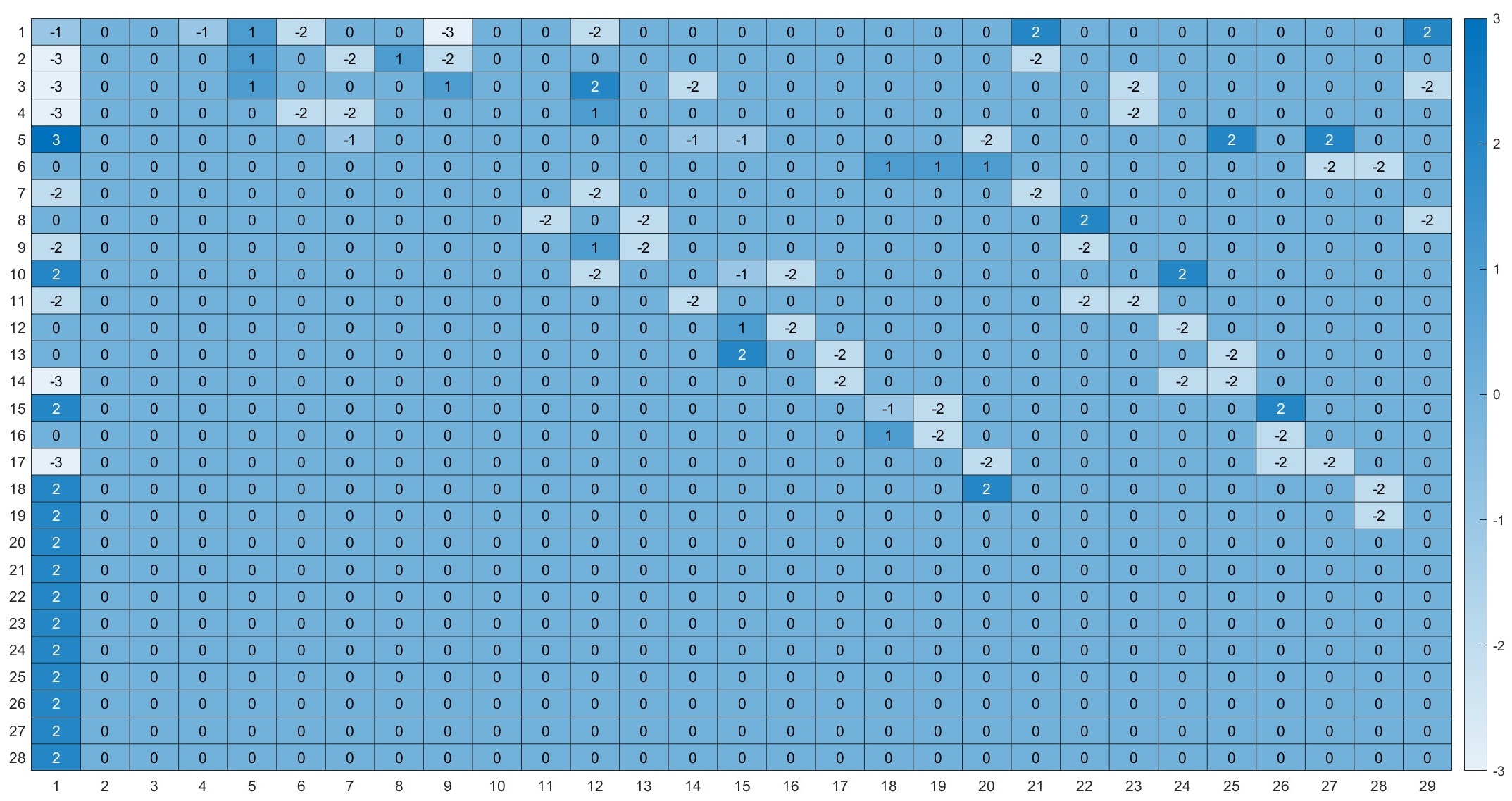}
\caption{Visualization of the local field vector $\mathbf{h}$ (first column) and the interaction matrix $\mathbf{J}$ (subsequent columns) for Problem 1. All the elements are integers from the range $[-3, 3]$.}
\label{fig:J0}
\end{figure}

\end{document}